\documentclass[conference]{IEEEtran}
\usepackage{float, amsmath, amsfonts, amssymb, amsthm, graphicx, tabularx, booktabs, ragged2e, epsfig, epstopdf, cite, subcaption, siunitx}

\newcolumntype{Y}{>{\RaggedRight\arraybackslash}X} 
\theoremstyle{plain}
\newtheorem{thm}{Theorem}
\newtheorem{lem}{Lemma}

\theoremstyle{definition}
\newtheorem{dfn}{Definition}

\newtheorem{asm}{Assumption}

\newcommand\given[1][]{\:#1\vert\:}
\newcommand\eqa{\mathrel{\overset{\makebox[0pt]{\mbox{\tiny\sffamily (a)}}}{=}}}
\newcommand\eqb{\mathrel{\overset{\makebox[0pt]{\mbox{\tiny\sffamily (b)}}}{=}}}

\newcommand\aprc{\mathrel{\overset{\makebox[0pt]{\mbox{\tiny\sffamily (c)}}}{\approx}}}
\newcommand\aprd{\mathrel{\overset{\makebox[0pt]{\mbox{\tiny\sffamily (d)}}}{\approx}}}

\newcolumntype{b}{X}

\newcounter{tempEquationCounter}
\newcounter{thisEquationNumber}

\graphicspath{{./figures/}}
\begin{document}
\title{A Blockage Model for the Open Area Mm-wave Device-to-Device Environment} 
\author
{\IEEEauthorblockN{Swapnil Mhaske\IEEEauthorrefmark{1}, Predrag Spasojevic\IEEEauthorrefmark{1}, Ahsan Aziz\IEEEauthorrefmark{2}}
\IEEEauthorblockA{\IEEEauthorrefmark{1}Wireless Information Network Laboratory, Rutgers University, New Brunswick, NJ, USA \\
Email:\{smhaske, spasojev\}@winlab.rutgers.edu}
\IEEEauthorblockA{\IEEEauthorrefmark{2}National Instruments Corporation, Austin, TX, USA \\
Email:\{ahsan.aziz\}@ni.com}}
\maketitle

\begin{abstract}
A significant portion of the \num{5}\textsuperscript{th} generation of wireless networks will operate in the mm-wave bands. One of the several challenges associated with mm-wave propagation is to overcome shadowing due to signal blockage caused by environmental objects. Particularly susceptible are nodes in a device-to-device network that typically operate at low power and in a blockage prone environment such as crowded open areas. In this work, we provide an insight into the effect of blockages on the signal quality for an open area device-to-device scenario. We propose a blockage model based on the homogeneous Poisson Point Process. The model provides the average signal attenuation as a soft metric that quantifies the extent of blockage. This not only indicates whether the signal is blocked but also measures how much the signal is attenuated due to one or more blockers. The analytical results are confirmed with the help of Monte Carlo simulations for real-world blocker placement in the environment.
\end{abstract}


\section{Introduction}
\label{sec:intro}
One of the differentiating characteristics of the next generation of wireless networks is its operation in the mm-wave bands. Technology developers aim to accomplish data rates of up to ten gigabits per second with an overall latency of less than 1ms \cite{raaf:5g}, \cite{cudak:5g} by leveraging the abundant bandwidth in the mm-wave spectrum \cite{khan:mmw_spectrum}. However, relative to the contemporary cellular spectrum, this comes with two major challenges. Firstly, mm-wave signals suffer from a greater propagation loss \cite{khan:mmw_spectrum}, \cite{rappaport:mmw_5g}. This can be addressed by beamforming to provide directional gains. Secondly, mm-waves are susceptible to penetration losses due to commonly found materials such as building materials, humans, vehicles and foliage \cite{rappaport:mmw_book}. Such a loss in the signal quality due to environmental objects is collectively referred to as the \emph{blockage} effect. Methods to alleviate this problem such as, beam adaptation and rerouting, relaying and non-line-of-sight (NLOS) communication are topics of active research. \\
\indent Link outage due to the blockage effect can impact several mechanisms such as beam setup, beamtracking, adaptive modulation, and coding [1]. Restoring the link quality may become prohibitively expensive in such scenarios, especially for the 5G target of \SI{1}{\milli\second} latency. To analyze outages due to blockages, it is crucial to model the blockage effect due to environmental objects. Studies to model blockages can be based on real world measurements (e.g. \cite{maccartney:hmn_blk}), ray tracing simulations (e.g. \cite{jacob:stc_hmn_blk_rtr}), and stochastic methods (applied in this work). Recent notable works that have developed stochastic models for blocker placement in the environment of operation are \cite{bai:blk_mod}, \cite{gapeyenko:anl_hmn_blk}. In \cite{bai:blk_mod}, authors use random shape theory to model the effect of randomly placed blockers in the environment. The line-of-sight (LOS) probability is derived and it is concluded that, by increasing the base station density the adverse effects of blockages can be alleviated. A detailed and insightful model for randomly distributed human presence around the receiver is provided in \cite{gapeyenko:anl_hmn_blk}. Authors propose a tractable model for human blockage and show that the blockage probability increases with the blocker population density and with the transmitter-receiver distance.  \\
\indent The presence of  blockers around the receiver directly depends on the spatial geometry of  the environment. For instance, human blockers may be found at random locations in a park. Whereas in the case of modeling blockages due to a street, humans are restricted to the sidewalk. It is not hard to see that a generic blockage model that can be applied to any environment is unrealizable. In this analysis, we focus on a scenario where blockers are randomly placed around the receiver in an open area environment. Such an environment is commonly found in application scenarios that fall under the 3GPP use case of massive Machine Type Communication (mMTC) \cite{std:tr_38913_v15_201806} and direct device-to-device (D2D) communication \cite{ansari:mmwave_d2d}, \cite{qiao:mmwave_d2d}, \cite{boccardi:disr_5Gtech} envisioned for 5G. Examples of such deployments are public safety networks (for LTE \cite{std:tr_36843_v12_201403}), warehouse and industrial robotics, wearable technology, and tactical networks. \\
\indent While most work in the area of modeling blockages focuses on determining the probability of LOS, we obtain a soft metric for the signal attenuation in the presence of blockages. We derive a tractable analytical upper bound on the signal attenuation due to blockages as a function of model parameters, namely, the spatial geometry, the blocker population, and the penetration loss of the blocker. To obtain this we leverage a simple stochastic geometry process: the two-dimensional homogeneous Poisson Point Process (PPP). The probability of link failure is derived based on the blockage model to provide modulation and coding design insights for operation in typical deployment scenarios. \\
\indent This paper is organized as follows. The spatial setup for modeling the blockage effect is described in Section \ref{sec:spl_set}. The derivation of the soft measure: the expected attenuation due to blockages is provided in Section \ref{sec:atn_blk}. The probability of outage as a function of the model parameters is subsequently derived in Section \ref{sec:otg_anl}. The theoretical and simulated performance along with insights into system design is the subject of Section \ref{sec:num_res_dis}. Concluding remarks on the work are presented in Section \ref{sec:conc}. \\
\indent Throughout the paper, random variables are denoted by upper case letters and the corresponding realizations are denoted by their lower case counterparts. $P_B(b)$ denotes the probability mass function (PMF) of the discrete random variable $B$. The cumulative distribution function (CDF) and the probability distribution function (PDF) of the continuous random variable $X$ is denoted by $F_X(x)$ and $f_X(x)$, respectively. $\mathcal{U}(a,b)$ denotes the uniform distribution over the interval $[a,b]$. $\phi(\cdot)$ denotes the PDF of the standard normal distribution, and $\Phi(\cdot)$ denotes its CDF. $\mathbb{E}(\cdot)$ is the expectation operator.

\section{Spatial Setup}
\label{sec:spl_set}
In this section we describe the spatial setup used to model blockages around the mm-wave link. To illustrate that the blocker presence can be approximated in a random manner, Fig. \ref{fig:spl_set} (left) depicts an open area D2D scenario. Such a scenario is typically found at town squares, transportation terminals, shopping malls, and parks.
\begin{figure*}
    \centering
        \includegraphics[width=\textwidth]{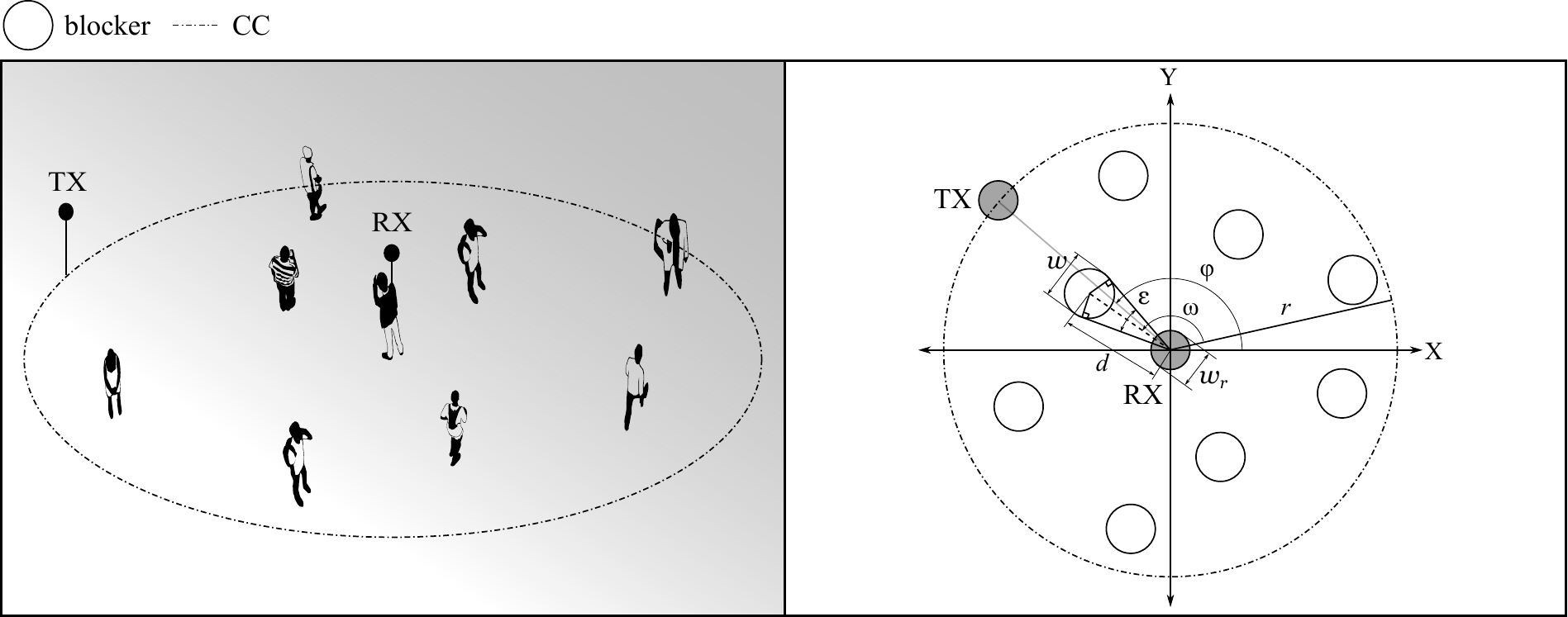}
        \caption{(Left) An example D2D scenario where the receiver (human with a mobile phone) associated with a transmitter (kiosk, sensor, or another mobile phone) is surrounded by human blockers. (Right) Azimuth plane geometry of the scenario on the left. Cylindrical blockers are randomly found inside the communication circle (CC) around the receiver which is located at the center of the CC. The receiver is blocked by a blocker along the direction $\phi$.}
    \label{fig:spl_set}
\end{figure*}

\begin{asm}
\label{asm:geo_sz}
	The size of the geometry under consideration is defined by the communication range of the receiver. This can be expressed by a circle of radius $r$ referred to as the communication circle (CC). As shown in Fig. \ref{fig:spl_set} (right) the receiver is at the center of the CC and the transmitter at its edge.
\end{asm}

To model the blockage effect we consider only those objects in the environment that can block the received signal in any direction. Such objects are termed as \emph{blockers}. Note that once the receiver is blocked in an arbitrary direction, it is assumed that it is fully blocked in the zenith. Thus, the height of the blocker is not modeled in the analysis that follows. We model blockers as cylinders \cite{jacob:hmn_blk_cyl} of diameter $w$. Letting $w_r$ denote the diameter of the cylindrical receiver and  $s \geq \frac{w+w_r}{2}$ denote the distance of the nearest blocker from the receiver. The location of a blocker is defined by the polar coordinates $(d,\omega) \in [s,r] \times (0, 2\pi]$. From \cite{lewis:ppp_sim} we know that for a homogeneous PPP inside the annular region $[s,r] \times (0, 2\pi]$, 
\begin{align}
	F_{D|R}(d \given r) &= \frac{d^2-s^2}{r^2-s^2} \quad d \in [s, r]
	\label{eq:cdf_D_R}
\end{align}
and $\Omega \sim \mathcal{U}(0,2\pi)$. The location of blockers in the environment is modeled by placing the blockers inside the CC around the receiver as per a two-dimensional homogeneous PPP \cite{haenggi:stoc_geom_book}, \cite{lewis:ppp_sim}. Whereas the blocker population is modeled with the intensity of the PPP $\rho$ which is the expected number of blockers per unit area. Thus, the average number of blockers within the CC is $\overline{\rho}(r)=\rho \times \pi r^2$. Also, the probability of having $M=m$ blockers around the receiver inside the CC is 
\begin{align}				
		P_{M|R}(m \given r)=\frac{\overline{\rho}^m(r)}{m!} e^{-\overline{\rho}(r)}, \quad m=0,1,\ldots, \infty
		\label{eq:pmf_M_R}
\end{align}

\section{Attenuation due to Blockages}
\label{sec:atn_blk}
Here, the focus is to determine the effect of blockages in terms of the attenuation caused to the received signal. Blocker material is characterized by the penetration loss that it can cause to a mm-wave signal. For instance, at \SI{15}{\giga\hertz} up to \SI{10}{\decibel} loss due to vehicles is reported in \cite{okvist:blk_chr}. Whereas, in \cite{maccartney:hmn_blk} the loss due to human blockage at \SI{73}{\giga\hertz} ranges from \SI{20}{\decibel} to \SI{40}{\decibel}. In this work we model the attenuation due to a blocker as its penetration loss. Based on the above, a closed form expression for the average attenuation along a given direction is presented in the following. \\
\indent We begin by analyzing the extent of blockage caused by a single blocker. From the geometry shown in Fig. \ref{fig:spl_set} (right) a blocker subtends angle $E=2 \arcsin(w/2D)$ as seen by the receiver at a distance $D$. This leads us to the following lemma. 
\begin{lem}
	For $2 \arcsin(w/2r) \leq \varepsilon \leq 2 \arcsin(w/2s)$,
	 \begin{align}
		f_{E|R}(\varepsilon \given r) &=\frac{w^2}{4 (r^2-s^2)} \frac{\cos(\varepsilon/2)}{\sin^3(\varepsilon/2)}.
	\label{eq:pdf_E_R}
	\end{align}
\end{lem}
\begin{proof}
Since $E=2 \arcsin(\frac{w}{2D})$, 
	\begin{align*}
		F_{E|R}(\varepsilon \given r) &= \Pr\left\lbrace 2 \arcsin\left(\frac{w}{2D}\right)\leq \varepsilon \given[\Big] R \leq r \right\rbrace \\
		&=1-F_{D|R}\left( \frac{w}{2 \tan(\varepsilon/2)} \given[\bigg] r \right) \\
		&\eqa 1-\frac{(w/2 \cdot s \cdot\sin(\varepsilon/2))^2-1}{(r/s)^2-1}
	\end{align*}
where, in (a) we used (\ref{eq:cdf_D_R}) and $2 \arcsin(w/2r) \leq \varepsilon \leq 2 \arcsin(w/2s)$. Differentiating with respect to $\varepsilon$ proves the lemma.
\end{proof}

\begin{dfn}[Cover]
\label{def:cvr}
The receiver is said to have a cover in the direction $\phi \in (0, 2\pi]$ if there exists a blocker located at $(d,\omega)$ such that $\omega-\varepsilon/2 \leq \phi \leq \omega+\varepsilon/2$, where $\varepsilon=2 \arcsin(w/2d)$.
\end{dfn}

\subsection{Probability of a Blockage}
\label{ssec:p_blk}
Consider that a blocker is located at $(d,\omega)$ in the CC of radius $r$. The probability that it blocks the received signal arriving from the direction $\phi$ is given by
\begin{align}
	h(r,\varepsilon) &= \Pr \left\lbrace \Omega-E/2 \leq \phi \leq \Omega+E/2 \given E=\varepsilon\right\rbrace \nonumber \\
	&\eqa \Pr \left\lbrace -E/2 \leq \Omega \leq E/2 \given E=\varepsilon \right\rbrace \nonumber \\
	&= F_\Omega(\varepsilon/2)-F_\Omega(-\varepsilon/2) \nonumber \\
	&\eqb \varepsilon/2\pi
	\label{eq:g_r_eps}
\end{align}
where, in (a) we have assumed $\phi=0^\circ$ since, the event $\left\lbrace \Omega-E/2 \leq \phi \leq \Omega+E/2 \right\rbrace$ is independent of $\phi$. Also, in (b) we have used the fact that $\Omega \sim \mathcal{U}(0,2\pi)$. Let $k_r:=w/2r$ and $k_s:=w/2s$. Then, the overall probability of a blocker in the CC of radius $r$ causing a cover along an arbitrary direction is
\begin{align}
	g(r) =& \int_{2\arcsin(k_r)}^{2\arcsin(k_s)} \! h(r,\varepsilon) f_{E|R}(\varepsilon \given r) \, \mathrm{d}\varepsilon \nonumber \\ 
		    \eqa{}& \frac{1}{2\pi} \int_{2\arcsin(k_r)}^{2\arcsin(k_s)} \! \varepsilon f_{E|R}(\varepsilon \given r) \, \mathrm{d}\varepsilon \nonumber \\
	\begin{split}
		\eqb{}& \frac{w^2}{8\pi (r^2-s^2)} \left[ \frac{2\arcsin(k_r)}{k_r^2} + 2 \left(\frac{1}{k_r^2}-1 \right)^{1/2} \right. \\
		&\left. - \frac{2\arcsin(k_s)}{k_s^2} - 2 \left(\frac{1}{k_s^2}-1 \right)^{1/2} \right].
	\end{split}
	\label{eq:g_r}
\end{align}
where, in (a) we have used (\ref{eq:g_r_eps}) and in (b) we have used (\ref{eq:pdf_E_R}). In the following theorem we provide the average attenuation caused to the received signal along $\phi$ due to multiple covers.

\subsection{Expected Attenuation}
\label{ssec:exp_atn}
Let $\zeta$ be the attenuation due to a blocker. Supposing that there are $M=m$ blockers in the CC, let $N=n \in [0,m]$ be the number of covers along the direction $\phi$. Based on the fact that the blocker locations $(D,\Omega)$ are i.i.d. for all the $m$ blockers we have,
\begin{align}
	P_{N|M,R}(n \given m,r) &= \binom{m}{n} g^n(r) \left[1-g(r)\right]^{m-n}.
\label{eq:p_nmr}
\end{align}
Let $A$ denote the attenuation of the received signal along $\phi$. Then, $A=\zeta^{N}$. 
\begin{thm}[Expected Attenuation]
Let the receiver be placed at the center of the CC of radius $r$. Let cylindrical blockers of diameter $w$ be placed around the receiver according to a two-dimensional homogeneous PPP with a sufficiently large intensity $\rho$. Then, the expected attenuation of a received signal arriving from the direction $\phi \in (0,2\pi]$ is given as 
	\begin{align*}
		\overline{A}(r) &:= \mathbb{E} \left[ A \given R=r \right] \approx e^{-\overline{\rho}(r)\left(1-e^{-g(r)(1-\zeta)} \right)}
	\end{align*}
	where, $\overline{\rho}(r)=\rho \times \pi r^2$ is the expected number of blockers inside the CC and $g(r)$ is the probability of a blocker causing a cover in the direction $\phi$.
\label{thm:exp_atn}
\end{thm}
\begin{proof}
The expected attenuation due to $m$ blockers inside the CC is given by,
	\begin{align}
		\mathbb{E}\left[A \given[\Big] M=m, R=r \right] &= \sum_{n=0}^{m} P_{N|M,R}(n \given m,r) \zeta^n.
	\end{align}
Applying the law of total probability we get,
\begin{align}
	\overline{A}(r) &= \sum_{m=0}^{\infty} \sum_{n=0}^{m} P_{N|M,R}(n \given m,r) \zeta^n P_{M|R}(m \given r) \nonumber \\
	&\eqb \sum_{m=0}^{\infty} \sum_{n=0}^{m} \binom{m}{n} g^n(r) \left(1-g(r)\right)^{m-n} \zeta^n \frac{\overline{\rho}^m(r)}{m!} e^{-\overline{\rho}(r)} \nonumber \\
	&\aprc \sum_{m=0}^{\infty} \sum_{n=0}^{m} \frac{\left(mg(r)\right)^n}{n!} e^{-mg(r)} \zeta^n \frac{\overline{\rho}^m(r)}{m!} e^{-\overline{\rho}(r)} \nonumber \\
	&= \sum_{m=0}^{\infty} e^{-mg(r)} \sum_{n=0}^{m} \frac{\left(\zeta m g(r) \right)^n}{n!} \frac{\overline{\rho}^m(r)}{m!} e^{-\overline{\rho}(r)} \nonumber \\
	&\aprd \sum_{m=0}^{\infty} e^{-mg(r)(1-\zeta)} \frac{{\overline{\rho}}^m(r)}{m!} e^{-\overline{\rho}(r)} \nonumber \\
	&= \sum_{m=0}^{\infty} \frac{\left(\overline{\rho}(r) e^{-g(r)(1-\zeta)}\right)^m}{m!} e^{-\overline{\rho}(r)} \nonumber \\
	&= e^{\left(\overline{\rho}(r) e^{-g(r)(1-\zeta)} \right)} e^{-\overline{\rho}(r)} \nonumber \\
	&= e^{-\overline{\rho}(r)\left(1-e^{-g(r)(1-\zeta)}\right)}
\label{eq:g_a}
\end{align}
where, in (b) we have used (\ref{eq:pmf_M_R}) and (\ref{eq:p_nmr}). In (c) we apply the Poisson limit theorem \cite{papoulis:prob_book} assuming sufficiently large values of $\rho$ (and consequently $m$) and sufficiently small values of $g(r)$ giving
$\binom{m}{n} g^n(r) [1-g(r)]^{m-n} \approx \frac{[mg(r)]^n}{n!} e^{-mg(r)}$. In (d) as a consequence of assuming sufficiently large values of $m$, we apply the Taylor series approximation to get $\sum_{n=0}^{m} \frac{[\zeta mg(r)]^n}{n!} \approx e^{\zeta mg(r)}$.
\end{proof}

\section{Numerical Results and Discussion}
\label{sec:num_res_dis}
The performance of the analytical results with respect to the Monte Carlo simulations is provided in this section. First, we discuss the attenuation result and subsequently assess its impact on the outage probability. We also provide an insight into the selection of the modulation and coding schemes (MCS) with respect to the model parameters. \\
\indent To determine the expected attenuation we perform a \num{e5} Monte Carlo simulation per $\rho$ value. In each trial, blockers are dropped as per the PPP and the resulting cover around the receiver is calculated. 
\begin{figure}
    \centering
        \includegraphics[width=\columnwidth]{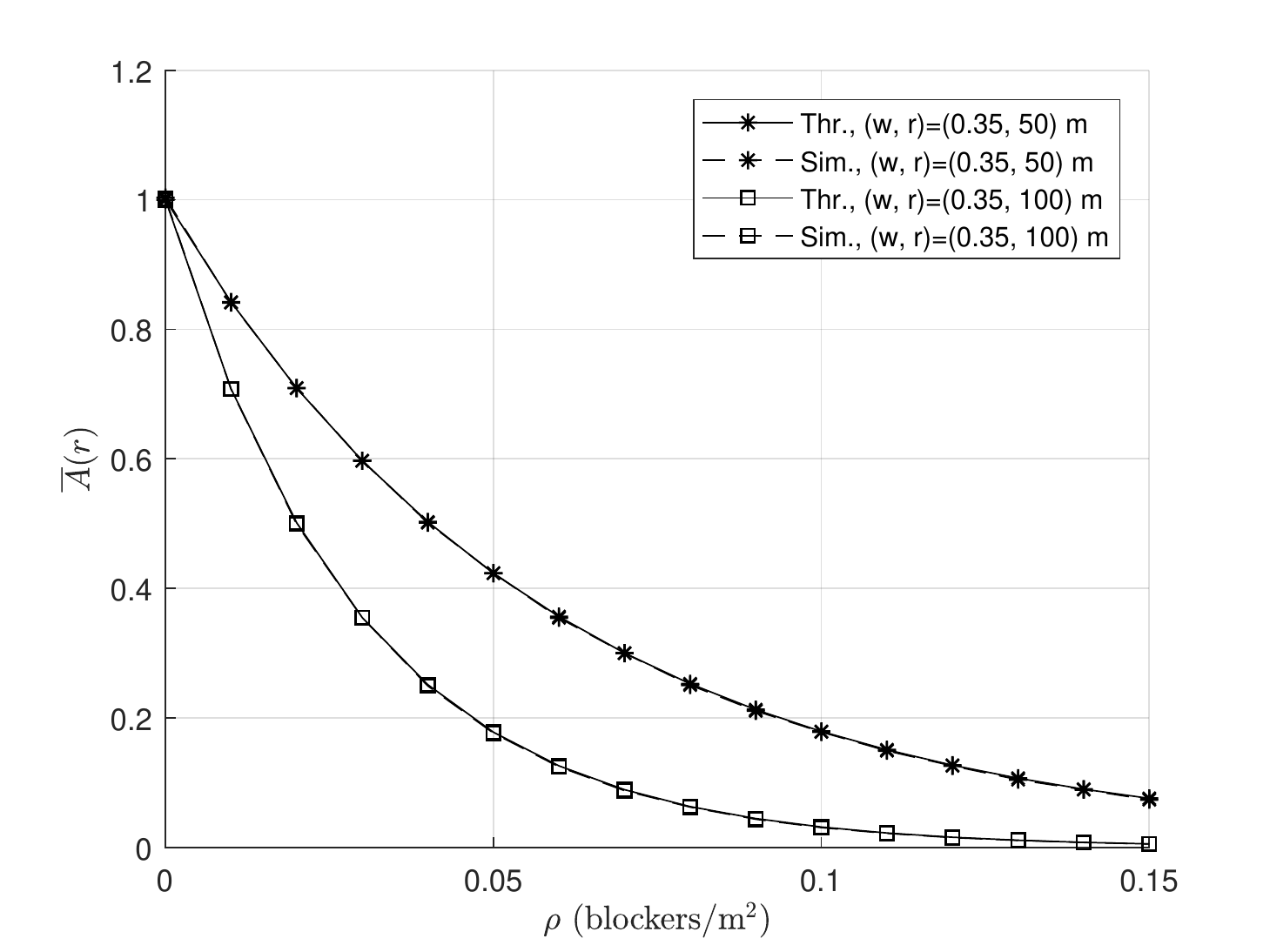}
        \caption{Comparison of the theoretical and simulated attenuation for various geometry sizes with $\zeta=\SI{-20}{\decibel}$.}
    \label{fig:ga}
\end{figure}
Fig. \ref{fig:ga} shows the theoretical and simulated values of the average attenuation for various geometry sizes ($r$) per blocker density. The attenuation $\zeta$ is set to a value of $\SI{-20}{\decibel}$ which is as per the human blockage measurements provided in \cite{maccartney:hmn_blk}. Note that, \cite{maccartney:hmn_blk} reports losses upto $\SI{-40}{\decibel}$, so the results shown here tend towards a best case loss scenario. The theoretical approximation closely follows the simulation. As expected, the attenuation worsens as the blocker population increases. 

\section{Conclusion}
\label{sec:conc}
The proposed blockage model quantifies the blockage effect in terms of a continuous metric: the signal attenuation that accounts for varying levels of blockage. The model holds for an open area mm-wave D2D environment. 

\section*{Acknowledgment}
The authors would like to thank the Department of Electrical \& Computer Engineering, Rutgers University for their continual support for this research work and our industry partner National Instruments Corporation for their valuable feedback and support.

\bibliographystyle{IEEEtran}
\bibliography{IEEEabrv,blk_mod_ba}

\end{document}